\documentclass[]{IEEEtran}

\usepackage[pdftex]{graphicx}  
\graphicspath{{./figures/}}
\DeclareGraphicsExtensions{.pdf,.jpeg,.png}                               
\usepackage{cite}
\usepackage{empheq}
\usepackage{amssymb}
\usepackage{bbm}
\usepackage{upgreek}

\usepackage{todonotes}
\usepackage{xspace}
\usepackage{wasysym}
\usepackage{mathrsfs}
\usepackage{enumitem}
\usepackage{epstopdf}
\usepackage{lipsum}
\usepackage[subnum]{cases}
\usepackage{booktabs}
\usepackage[overlay,absolute]{textpos}

\newcommand{\real}{\ensuremath{\mathbb{R}}}
\newcommand{\nat}{\ensuremath{\mathbb{N}}}
\newcommand{\Zge}{\ensuremath{\mathbb{Z}_{\ge 0}}}
\newcommand{\Ba}{\ensuremath{\mathbb{B}}}

\newcommand{\Bu}{{B\"uchi \xspace}}
\newcommand{\Au}{\ensuremath{\mathscr{A}}}
\newcommand{\Oh}{\ensuremath{\mathscr{O}}}

\newcommand{\BA}{\ensuremath{\textup{BA}}}

\newcommand{\DBAC}{\ensuremath{D_\textup{BA}^\textup{C}}}
\newcommand{\GBAC}{\ensuremath{G_\textup{BA}^\textup{C}}}
\newcommand{\deltac}{\ensuremath{\delta^\textup{C}}}
\newcommand{\OC}{\ensuremath{O^\textup{C}}}
\newcommand{\HyBAC}{\ensuremath{\mathcal{H}_\textup{BA}^\textup{C}}}
\newcommand{\OBAC}{\ensuremath{\mathcal{O}_\textup{BA}^\textup{C}}}

\newcommand{\VBA}{\ensuremath{V_\textup{BA}}}

\newcommand{\Vu}{\ensuremath{V_\textup{u}}}
\newcommand{\Vl}{\ensuremath{V_\textup{l}}}
\newcommand{\muBA}{\ensuremath{\mu_\textup{BA}}}

\renewcommand{\sf}{\ensuremath{s_\textup{f}}}
\newcommand{\Sf}{\ensuremath{S_\textup{f}}}

\newcommand{\Sso}{\ensuremath{\smat{s\\o}}}
\newcommand{\Ind}{\ensuremath{\mathbb{I}}}
\newcommand{\Ap}{\ensuremath{\mathbf{A}}}
\newcommand{\Bp}{\ensuremath{\mathbf{B}}}
\newcommand{\Cp}{\ensuremath{\mathbf{C}}}

\newcommand{\Lp}{\ensuremath{\mathbf{L}}}
\newcommand{\Kp}{\ensuremath{\mathbf{K}}}
\newcommand{\Fp}{\ensuremath{\mathbf{F}}}
\newcommand{\gp}{\ensuremath{\mathbf{g}}}
\newcommand{\Pp}{\ensuremath{\mathbf{P}}}
\newcommand{\Qp}{\ensuremath{\mathbf{Q}}}

\newcommand{\smat}[1]{\ensuremath{\left[\begin{smallmatrix}#1\end{smallmatrix}\right]}}
\newcommand{\bmat}[1]{\ensuremath{\begin{bmatrix}#1\end{bmatrix}}}
\newcommand{\Hy}{\ensuremath{\mathcal{H}}}
\newcommand{\HyR}[1]{\ensuremath{\left.\mathcal{H}\right|_{#1}}}
\newcommand{\RnmO}{\ensuremath{\real^n \backslash \O}}

\newcommand{\So}{\ensuremath{\mathcal{S}}}
\newcommand{\Cr}{\ensuremath{{C}_\textup{r}}}
\newcommand{\Dr}{\ensuremath{{D}_\textup{r}}}
\newcommand{\phir}{\ensuremath{{\phi}_\textup{r}}}

\newcommand{\lambdac}{\ensuremath{\lambda_\textup{c}}}
\newcommand{\lambdad}{\ensuremath{\lambda_\textup{d}}}

\newcommand{\dMAX}{\ensuremath{d_\textup{MAX}}}
\newcommand{\wmin}{\ensuremath{w_\textup{min}}}

\newcommand{\xh}{\ensuremath{x_\textup{H}}}
\newcommand{\xhDot}{\ensuremath{\dot{x}_\textup{H}}}
\newcommand{\Ch}{\ensuremath{C_\textup{H}}}
\newcommand{\Dh}{\ensuremath{D_\textup{H}}}
\newcommand{\fh}{\ensuremath{f_\textup{H}}}
\newcommand{\Gh}{\ensuremath{G_\textup{H}}}

\newcommand{\muh}{\ensuremath{\mu_\textup{H}}}
\newcommand{\Vh}{\ensuremath{V_\textup{H}}}

\renewcommand{\O}{\ensuremath{\mathcal{O}}}
\renewcommand{\Oh}{\ensuremath{\mathcal{O}_\textup{H}}}
\newcommand{\df}{\ensuremath{\mathfrak{d}}}

\newcommand{\toto}{\ensuremath{\rightrightarrows}}

\newcommand{\UGR}{UGR\xspace}
\newcommand{\UlyGR}{UGR\xspace}

\DeclareMathOperator{\dom}{dom}

\DeclareMathOperator{\gph}{gph}

\newtheorem{assumption}{Assumption}

\newtheorem{proposition}{Proposition}
\newtheorem{lemma}{Lemma}

\newtheorem{definition}{Definition}
\newenvironment{proof}{\noindent\textit{Proof.}}{\hfill $\square$}

\newcounter{myremcounter}
\newenvironment{remark}{
\refstepcounter{myremcounter}
\noindent\textit{Remark \themyremcounter:}}{\hfill $\triangle$}

\begin{document}
\title{Satisfaction of linear temporal logic specifications\\
through recurrence tools for hybrid systems}
\author{Andrea~Bisoffi, Dimos V. Dimarogonas% <-this % stops a space
\thanks{This work was supported in part by the Swedish Research Council (VR), the Swedish Foundation for Strategic Research (SSF), the Knut and Alice Wallenberg Foundation (KAW), the SRA ICT TNG project STaRT, and the European Research Council (ERC) through ERC StG BUCOPHSYS.}% <-this % stops a space
\thanks{Andrea~Bisoffi, Dimos V. Dimarogonas are with the Division of Decision and Control Systems, KTH Royal Institute of Technology, SE-100 44 Stockholm, Sweden \texttt{\{bisoffi, dimos\}@kth.se}}% <-this % stops a space
}

\maketitle

\begin{abstract}
In this work we formulate the problem of satisfying a linear temporal logic formula on a linear plant with output feedback, through a recent hybrid systems formalism. We relate this problem to the notion of recurrence introduced for the considered formalism, and we then extend Lyapunov-like conditions for recurrence of an open, unbounded set. One of the proposed relaxed conditions allows certifying recurrence of a suitable set, and this guarantees that the high-level evolution of the plant satisfies the formula, without relying on discretizations of the plant. Simulations illustrate the proposed approach.
\end{abstract}

\IEEEpeerreviewmaketitle

\begin{textblock*}{\textwidth}(17mm,4mm)%
\small\bfseries{\color{red}\noindent\textcopyright \textcopyright 2020 IEEE. Personal use of this material is permitted. Permission from IEEE must be obtained for all other uses, in any current or future media, including reprinting/republishing this material for advertising or promotional purposes, creating new collective works, for resale or redistribution to servers or lists, or reuse of any copyrighted component of this work in other works.}
\end{textblock*}
\section{Introduction}

Linear temporal logic (LTL, see, e.g., \cite{baier2008principles,belta2017formal}) provides a tool to formulate richly expressive control specifications for continuous-time plants, such as high-level tasks for multi-robot systems. An LTL formula can be equivalently translated into a \Bu automaton (BA) \cite[Thm.~5.41]{baier2008principles}, and we thus consider the equivalent BA instead of the LTL formula throughout this work. Then, the combination of the BA and the continuous-time plant can be appealingly addressed through a hybrid systems framework \cite{goebel2012hybrid}.

As a key property, a word of atomic propositions satisfies the LTL formula if the sequence of states induced by the word in the corresponding BA visits some accepting states infinitely often. This property has an intriguing relation to the recurrence property for hybrid systems in~\cite{subbaraman2016equivalence}. The study of recurrence for hybrid systems was initiated by its relevance in the case of stochastic hybrid systems~\cite{teel2013lyapunov}. It was specialized in~\cite{subbaraman2016equivalence} for the nonstochastic case where global recurrence of an open, bounded set is shown to be equivalent to the existence of a smooth Lyapunov-like function relative to that set. We emphasize that recurrence of a set does not entail forward invariance or stability of such set, but is an attractivity-like property. Together with completeness of solutions, it matches well the acceptance condition of the LTL formula for the BA, as shown in this paper.
A similar notion of recurrence was also studied in~\cite{fiacchini2018stabilization} (recurrent stabilizability of language-constrained discrete-time linear switching systems).

In this work we show the relevance of recurrence for hybrid systems in the context of LTL specifications on a linear continuous-time plant with output feedback where we want the LTL specification to enforce in an automated way which regions of interest should be reached and in which order. 
First, the framework in~\cite{goebel2012hybrid} is suitable here because, at the same time, \textit{(i)}~it allows leveraging computationally efficient control laws for the continuous-time part kept as is, \textit{(ii)}~it provides Lyapunov-like tools for sets such as the considered regions of interest, \textit{(iii)}~it captures set-valued dynamics such as the nondeterministic BA corresponding to the LTL formula.
Moreover, the relevant case of output feedback (instead of full state feedback) motivates us to extend the results given for open, bounded sets in~\cite{subbaraman2016equivalence}. 
Indeed, regions of interest defined in the \emph{output} variables induce open, \emph{unbounded} sets in the state variables. We then show that for open, unbounded sets, the sufficiency of the Lyapunov-like result~\cite[Thm.~5]{subbaraman2016equivalence} for recurrence still holds (Lemma~\ref{lem:recOnotBounded}) and we provide a relaxed Lyapunov-like condition (Proposition~\ref{prop:recOnotBoundedRelaxed}) needed for the proof of our main result (Proposition~\ref{prop:OhRec}). 
Our designed hybrid scheme and the certificates of recurrence in terms of hybrid Lyapunov-like functions guarantee then satisfaction of an LTL formula when the designed high-level hybrid scheme is endowed with obstacle avoidance low-level controllers, which we do not pursue here.
Finally, our approach allows leveraging continuous-time control laws for the plant, and is an alternative to approaches which discretize the continuous-time plant into a transition system and use automata-based tools to find a control strategy (see \cite{tabuada2009verification}, \cite{baier2008principles}, \cite{belta2017formal} and references therein). Such approaches suffer from the computational cost induced by a possibly very large discretization of the plant into a transition system, which is avoided here altogether.
Approaches using framework \cite{goebel2012hybrid} for the satisfaction of LTL formulae were presented in~\cite{bisoffi2018hybrid} (for a fragment of LTL, i.e., syntactically co-safe LTL, and state feedback) and in~\cite{han2018sufficient} (sufficient conditions for single temporal operators of LTL formulae).

The main contribution of this paper shows that the satisfaction of an LTL formula can be certified through Lyapunov-like tools, once it is reframed in terms of recurrence of a suitable set. As a second contribution, we nontrivially extend Lyapunov-like sufficient conditions for a case not covered in~\cite{subbaraman2016equivalence}, motivated by output feedback.

The paper is structured as follows. Section~\ref{sec:prel} introduces preliminaries. The relation between the satisfaction of LTL and the recurrence property is in Section~\ref{sec:BA}. Section~\ref{sec:addCT} presents the hybrid dynamics of BA and continuous-time plant, the auxiliary Lyapunov-like conditions and the main result of the satisfaction of LTL in terms of recurrence. Section~\ref{sec:ex} provides a numerical example. All proofs are in the appendix.

\textit{Notation:} 
$\real$, $\real_{\ge 0}$, $\nat$ are the sets of reals, nonnegative reals, nonnegative integers.
For $w_1 \in \real^{n_1}$ and $w_2 \in \real^{n_2}$, $(w_1,w_2):=\smat{w_1^T & w_2^T}^T$. 
For a set-valued mapping $M\colon \real^n \rightrightarrows \real^n$, the domain of $M$ is $\dom M := \{ x \in \real^n  \colon M(x) \neq \emptyset \}$ and its graph is the set $\gph M:= \{(x,y)\in \real^n \times \real^n \colon y \in M(x)\}$. 
$\Zge$ denotes the set of nonnegative integer numbers. 
$\wedge$, $\vee$, $\implies$ denote the logical conjunction, disjunction, implication. 
For $c\in \real^n$ and $r>0$, $\Ba(c,r)$ denotes the closed ball with center $c$ and radius $r$.
For a set $S\subset \real^n$, $S^\circ$, $\partial S$ and $\overline S$ denote its interior, boundary and closure. 
For $v \in \real^n$ and a set $S \subset \real^n$, the indicator function $\Ind_{S}$ is defined as $\Ind_{S}(v) := 1$ if $v \in S$ and as $\Ind_{S}(v) := 0$ if $v \notin S$.

\section{Preliminaries}
\label{sec:prel}

We consider hybrid systems $\Hy$ of the class \cite{goebel2012hybrid}:
\begin{subequations}
\label{eq:hs}
\begin{empheq}[left={\Hy:\empheqlbrace}]{align}
& \dot x  \in F(x), & & x \in C \label{eq:hsFlow} \\
& x^+  \in G(x), & & x \in D, \label{eq:hsJump}
\end{empheq}
\end{subequations}
denoted briefly as $\Hy = (F,C,G,D)$. We make the next mild assumption on $\Hy$.
\begin{assumption}
\label{ass:hbc}
The data $(F,C,G,D)$ of $\Hy$ in~\eqref{eq:hs} satisfy the hybrid basic conditions in~\cite[Assumption~6.5]{goebel2012hybrid}, i.e., 
\begin{itemize}
\item $C\subset \dom F$ and $D\subset \dom G$ are closed sets in $\real^n$;
\item the set-valued mappings $F$ and $G$ have a closed graph and are locally bounded  relative to $C$ and $D$, respectively;
\item $F(x)$ is convex for each $x \in C$.
\end{itemize}
\end{assumption}
A set $E \subset \real_{\ge 0} \times \nat$ is a \emph{hybrid time domain} if it is a union of a finite or infinite sequence of intervals $[t_j,t_{j+1}]\times \{j\}$, with the last interval (if existent) possibly of the form $[t_j,T)$ with $T$ finite or $T=\infty$ \cite[Def.~2.3]{goebel2012hybrid}.
A function $\phi \colon E \to \real^n$ is a \emph{hybrid arc} if $E$ is a hybrid time domain and if for each $j \in \nat$, the function $t \mapsto \phi(t,j)$ is locally absolutely continuous on $\{t \colon (t,j) \in E\}$ \cite[Def.~2.4]{goebel2012hybrid}.
Given a hybrid arc $\phi$, $\dom \phi$ represents its domain, which is a hybrid time domain and for which the operations $\sup_t$, $\sup_j$, $\sup$ are defined in~\cite[p.~27]{goebel2012hybrid}.
A hybrid arc $\phi$ is a \emph{solution to $\Hy$} if: \textit{(i)} $\phi(0,0) \in C \cup D$; \textit{(ii)} for every $j \in \nat$, $\phi(t,j) \in C$ and $\dot \phi(t,j) \in F(\phi(t,j))$ for almost all $t \in I^j:= \{t \colon (t,j) \in \dom \phi \}$; \textit{(iii)} for every $(t,j) \in \dom \phi$ such that $(t,j+1) \in \dom \phi$, $\phi(t,j) \in D$ and $\phi(t,j+1) \in G(\phi(t,j))$ \cite[p.~124]{goebel2012hybrid}. 
A solution is \emph{maximal} if it cannot be extended \cite[Def.~2.7]{goebel2012hybrid}, and \emph{complete} if its domain is unbounded (in the $t$- or $j$-direction) \cite[p.~30]{goebel2012hybrid}.
For a set $U$, $\So_\Hy(U)$ denotes the set of all maximal solutions $\phi$ to $\Hy$ with $\phi(0,0)\in U$.
We define the restriction $\HyR{\Gamma}$ of the hybrid system $\Hy=(F,C,G,D)$ to the set $\Gamma$ as:
\begin{subequations}
\label{eq:hsRestr}
\begin{empheq}[left={\HyR{\Gamma}:\empheqlbrace}]{align}
& \dot x  \in F(x), & & x \in C \cap \Gamma\label{eq:hsFlowRestr} \\
& x^+  \in G(x), & & x \in D \cap \Gamma. \label{eq:hsJumpRestr}
\end{empheq}
\end{subequations}
We recall next \cite[Def.~2]{subbaraman2016equivalence}, where no finite escape times for~\eqref{eq:hsFlow} means that there exist no solutions to~\eqref{eq:hsFlow} escaping to infinity at a finite time.
\begin{definition}[{\hspace{1sp}\cite[Def.~2]{subbaraman2016equivalence}}]
\label{def:GRandUGR}
A set $\O \subset \real^n$ is \emph{uniformly globally recurrent} (UGR) for $\Hy$ if \emph{(i)}~there are no finite escape times for~\eqref{eq:hsFlow}, and \emph{(ii)}~for each compact set $K$, there exists $T > 0$ such that for each solution  $\phi \in \So_\Hy (K)$, either $t +j < T$ for all $(t, j) \in \dom \phi$ or there exists $(t, j) \in \dom \phi$ such that $t + j \le T$ and $\phi(t, j) \in \O$.
\end{definition}
Intuitively speaking, item \textit{(ii)} of Definition~\ref{def:GRandUGR} asks that, uniformly over compact sets, all solutions either stop, or hit  $\O$.

\section{\Bu automaton as a difference inclusion and recurrence}
\label{sec:BA}

A generic linear temporal logic (LTL) formula can be translated into an equivalent nondeterministic \Bu automaton (BA) as follows.
\begin{definition}[{\hspace{1sp}\cite[Defs.~2.5-2.6]{belta2017formal}}]
A (nondeterministic) \emph{\Bu automaton} (BA) is a tuple $\Au = (S, S_0, O, \delta, \Sf)$, where 
$S\subset \Zge$ is a finite set of states taken as nonnegative integers,
$S_0 \subset S$ is a set of initial states,
$O\subset \Zge$ is a finite set of observations taken as nonnegative integers,
$\delta \colon S \times O \toto S$ is a nondeterministic transition function, $\Sf \subset S$ is a set of accepting states. 
The \emph{semantics} of a BA are defined over infinite words of observations in $O^\omega$. 
A run of $\Au$ over an infinite word of observations $w_O = w_O(1)w_O(2)\dots \in O^\omega$ is a sequence $w_S = w_S(1)w_S(2)\dots\in S^\omega$ where $w_S(1) \in S_0$ and $w_S(k + 1) \in \delta(w_S(k),w_O(k))\subset S$ for all $k \ge  1$. 
A word $w_O$ is \emph{accepted} by $\Au$ if there exists at least one run $w_S$ over $w_O$ that visits $\Sf$ infinitely often, i.e., the intersection with $\Sf$ of the states appearing in the run $w_S$ infinitely often is nonempty.
\label{def:BA}
\end{definition}

In a formal-methods setting, the sets $S$ and $O$ are sets of labels and atomic propositions, which can be indexed by nonnegative integers. In Definition~\ref{def:BA}, we directly identify the sets $S$ and $O$ with such indices of their elements, with some abuse of the standard notation.

We make the next assumption that we are given a feasible LTL formula, which corresponds to the existence of at least one accepting state $\sf \in \Sf$ that can be visited infinitely often.
\begin{assumption}
\label{ass:LTLfeas}
For the BA $\Au = (S, S_0, O, \delta, \Sf)$, there exists $\sf \in \Sf$ reachable from some initial state $s_0 \in S_0$ and containing a cyclic path through itself. Without loss of generality, we remove all states in $S$ from which an accepting $\sf$ containing a cyclic path through itself cannot be reached.
\end{assumption}

Let us interpret the BA $\Au$ in Definition~\ref{def:BA} as a discrete-time dynamical system with state $s$ and driven by $o$:
\begin{align}
\label{eq:bareBA}
s^+ & \in \delta(s,o)
\end{align}
where $\delta$ is in Definition~\ref{def:BA}, and is generally set-valued since the BA is nondeterministic, but can be empty, e.g., in the case when there exists no outgoing transition labelled $o$ from the logical state $s$. Then, the only observations that can be effectively taken from $s$ correspond to the set (indexed by $s$)
\begin{equation}
\label{eq:Os}
O_s := \{ o \in O \colon \delta(s,o) \neq \emptyset \}.
\end{equation}
Thanks to Assumption~\ref{ass:LTLfeas}, we have $O_s \neq \emptyset$ for each $s\in S$ (otherwise such $s$ would have been removed).
\begin{subequations}
\label{eq:HyGBA}
The bare evolution of the BA can then be expressed through the state $\chi := \Sso$ as the constrained difference inclusion
\begin{align}
& \chi^+ \in G_\BA( \chi),\qquad \chi \in D_\BA\\
& \,\,\,\,G_\BA ( \chi )\!= \!G_\BA ( \smat{s\\o} )  := \big\{\! \smat{s'\\o'}\! \colon\! s' \in \delta(s,o), o' \in O_{s'}\! \big\}\\
\label{eq:DBA}
&\,\,\,\, D_\BA  := \big\{ \smat{s \\ o } \colon s \in S, o \in O_s \big\}.
\end{align}
\end{subequations}

Let us now introduce shortest-path distances on the BA.
The BA can be seen as a digraph where each $s$ represents a vertex, and the edges from $s$ to any element of $\delta(s,O)$ are labelled by observations in $O$. We compute then for each node $s\in S$ its shortest-path distance $\hat d$ to any other node $\sf\in \Sf$ as
\begin{subequations}
\label{eq:distOnAutomaton}
\begin{equation}
\label{eq:hatDist}
\begin{aligned}
&(s,\sf)\mapsto \hat d(s,\sf)\\
& := \! \begin{cases}
\infty, \text{\hspace{1.5cm} \textit{if} there is no path from $s$ to $\sf$}\\
\begin{minipage}{6.7cm}
minimum number of edges \\
in any path from $s$ to $\sf$, \hspace{1.4cm} \textit{otherwise}
\end{minipage}
\end{cases}
\end{aligned}
\end{equation}
by breadth-first search algorithm \cite[\S 22.2]{cormen2009introduction}. By~\eqref{eq:hatDist}, we define the distance of $s \in S$ to the set of accepting states $\Sf$ as
\begin{equation}
\label{eq:dist}
d(s,\Sf) := \min_{\sf \in \Sf} \hat d(s,\sf),
\end{equation}
\end{subequations}
which is the minimum shortest-path distance from $s$ over the accepting states $\Sf$.\footnote{Whereas $\hat d$ can take $\infty$ as a value, Assumption~\ref{ass:LTLfeas} excludes that $d$ takes $\infty$ as a value, otherwise no accepting state $\sf$ could be reached from $s$ and such $s$ would have been removed from the BA.} Introduce now the set-valued mapping
\begin{equation}
\label{eq:deltac}
\begin{aligned}
\deltac (s& , o) :=  \{ s' \in \delta(s,o) \colon \\
& \big(s\notin \Sf \implies d(s', \Sf) < d(s,\Sf) \big) \\
& \quad \wedge \big(s\in \Sf \implies d(s',\Sf) = \min_{\df \in \delta(s,O)} d(\df,\Sf) \big) \},
\end{aligned}
\end{equation}
whose conditions are interpreted as follows. If $s$ is not an accepting state, $s'\in\delta(s,o)$ is chosen to decrease the distance to $\Sf$. If $s$ is an accepting state, $s' \in \delta(s,o)$ is chosen so to increase the distance to the accepting states as little as possible. These two properties of $\deltac$ are beneficial in the Lyapunov-like conditions for recurrence used in the sequel.

We then further constrain the BA in~\eqref{eq:HyGBA} using $\deltac$.
To this end, define for each $s \in S$ the next subset of $O_s$ as
\begin{equation}
\label{eq:OC_s}
\OC_s:= \{ o \in O_s \colon \deltac (s,o) \neq \emptyset \},
\end{equation}
which has the next relevant property.

\begin{lemma}
\label{lem:OCs nonempty}
Under Assumption~\ref{ass:LTLfeas}, $\OC_s \neq \emptyset$ for each $s \in S$.
\end{lemma}

Then, the further constrained difference inclusion reads
\begin{subequations}
\label{eq:HyGBAC}
\begin{align}
& \HyBAC :\, \chi^+ \in \GBAC( \chi),\qquad \chi \in \DBAC \label{eq:hsBAC}\\
& \hspace*{1.3cm} \GBAC ( \chi )\! := \big\{\! \smat{s'\\o'}\! \colon\! s' \in \deltac(s,o), o' \in \OC_{s'}\! \big\} \label{eq:GBAC}\\
& \hspace*{1.3cm} \DBAC  := \big\{\! \smat{s \\ o } \!\colon s \in S, o \in \OC_s \big\}. \label{eq:DBAC}
\end{align}
\end{subequations}
which is to be compared to~\eqref{eq:HyGBA}. 
$\HyBAC$ enjoys the next properties.
\begin{lemma}
\label{lem:hbc-completeness-GBAC}
Under Assumption~\ref{ass:LTLfeas}, $\HyBAC$ satisfies Assumption~\ref{ass:hbc} (in particular 
$\DBAC \subset \dom \GBAC$), and
$\GBAC(\DBAC) \subset \DBAC$.
\end{lemma}

Consider the open, bounded set
\begin{equation}
\label{eq:OBAC}
\OBAC:= \{ (\sf , o ) + r \colon \sf \in \Sf, o \in \OC_{\sf}, r \in
\Ba(0,\tfrac{1}{3})^\circ\}.
\end{equation}
We consider $\OBAC$ instead of $\{(\sf,o) \colon \sf \in \Sf,\, o \in \OC_{\sf} \}$ because we adhere to the setting of~\cite{subbaraman2016equivalence}, which characterizes recurrence for open sets. However, we emphasize that $\OBAC$ and the latter set are the same for our purposes. Indeed, the radius in the inflated set $\OBAC$ is less than one, $s$ and $o$ take values in the integers by~\eqref{eq:HyGBAC} and Definition~\ref{def:BA}, so the values introduced by the inflation are nonintegers and \emph{artificial}. 
$\OBAC$ has the next property.
\begin{lemma}
\label{lem:recOBAC}
Under Assumption~\ref{ass:LTLfeas}, the set $\OBAC$ in~\eqref{eq:OBAC} is \UlyGR for $\HyBAC$ in~\eqref{eq:HyGBAC}.
\end{lemma}

Each maximal solution $\phi$ to~\eqref{eq:HyGBAC} is complete by Lemma~\ref{lem:hbc-completeness-GBAC}, and reaches infinitely often $\OBAC$ by Lemma~\ref{lem:recOBAC}. Indeed, the existence of a \emph{single}, finite $j\in \dom \phi$ satisfying $\phi(j) \in \OBAC$ (from Definition~\ref{def:GRandUGR} of \UGR of $\OBAC$) is sufficient to imply that $\phi$ reaches infinitely often $\OBAC$. Therefore, by the acceptance condition of Definition~\ref{def:BA}, for each complete solution $\phi=(s,o)$ to $\HyBAC$, the word corresponding to the component $o$ is accepted by $\Au$ by construction of $\OBAC$.

Imposing the strict decrease in the distance to $\Sf$ (see \eqref{eq:deltac}, \eqref{eq:GBAC}-\eqref{eq:DBAC}) allows giving guarantees in terms of recurrence, but can prune away some solutions. However, at least one solution with strict decrease exists and is found by our approach.

\section{Continuous-time dynamics and hybrid system}
\label{sec:addCT}

The continuous-time plant is given by a linear system with state $\xi \in \real^\nu$, control $u \in \real^m$ and output $y \in \real^p$:
\begin{equation}
\label{eq:ctPlant}
\dot \xi = \Ap \xi + \Bp u, \quad y = \Cp \xi.
\end{equation}
As we explained in Section~\ref{sec:BA}, our ultimate goal is to generate a word of observations $w_O$ (see Definition~\ref{def:BA}) that is accepted by the BA $\Au$. Then, we need to specify how to associate each solution to~\eqref{eq:ctPlant} with a word of observations. 
For each observation $o$, consider an open set $Y_o$ as the region of interest for $o$.
We say that the \emph{solution has generated the observation} $o$ when, under a suitable control action $u$, the output $y$ belongs to $Y_o$ and a jump is enabled in the \emph{preliminary} hybrid system
\begin{equation}
\label{eq:hsJustPlant}
\begin{aligned}
 \dot \xi & = \Ap \xi + \Bp u , & & y=\Cp \xi \in {\real^p \backslash {Y_o}}\\
 \xi^+ & = \xi, & & y=\Cp \xi \in Y_o,
\end{aligned}
\end{equation}
which we further design in Section~\ref{sec:recSetHybSys}. Compared to~\cite{bisoffi2018hybrid}, we consider here the more realistic setting of output feedback, instead of full-state feedback. This has the implication that even an open, \emph{bounded} set $Y_o$ would result into an open, \emph{unbounded} set for the state $\xi$ (whenever the matrix $\Cp$ in~\eqref{eq:ctPlant} has, nontrivially, less rows than columns). 
As a main result of this section (Proposition~\ref{prop:OhRec}), \UGR of a suitable set $\Oh$ (defined below in~\eqref{eq:Oh}) guarantees satisfaction of the LTL formula. However, $\Oh$ is open but \emph{unbounded} in the meaningful case of output feedback, and recurrence (and the corresponding Lyapunov tools) are provided in~\cite{subbaraman2016equivalence} for open, \emph{bounded} sets. So, we extend in Section~\ref{sec:auxResRec} some of those results for open, unbounded sets.

\subsection{Results for recurrence of an open, unbounded set}
\label{sec:auxResRec}

We show in Lemma~\ref{lem:recOnotBounded} that the (Lyapunov-like) sufficient conditions for \UGR of an open, bounded set $\O$ in~\cite{subbaraman2016equivalence} remain valid for \UGR of an open, unbounded set $\O$. 
\begin{lemma}
\label{lem:recOnotBounded}
Let the hybrid system $\Hy$ in~\eqref{eq:hs} satisfy Assumption~\ref{ass:hbc}, $\O \subset \real^n$ be an open, unbounded set and $V \colon \real^n \to \real_{\ge 0}$ a smooth function, radially unbounded relative to $C \cup D$, for which there exists $\mu>0$ such that
\begin{subequations}
\label{eq:LyapFunRec}
\begin{align}
& \langle \nabla V(x), f \rangle \le -1 + \mu \Ind_\O (x) & & \forall x \in C, f \in F(x) \label{eq:LyapFunRecFlow} \\
& V(g) - V(x) \le -1 + \mu \Ind_\O (x) & & \forall x \in D, g \in G(x). \label{eq:LyapFunRecJump}
\end{align}
\end{subequations}
Then $\O$ is \UlyGR for $\Hy$.
\end{lemma}

Lemma~\ref{lem:recOnotBounded} and its proof are instrumental to prove the next Proposition~\ref{prop:recOnotBoundedRelaxed}, where we propose relaxed Lyapunov-like conditions for \UGR with respect to Lemma~\ref{lem:recOnotBounded} in the same way that \cite[Prop.~3.29]{goebel2012hybrid} proposed relaxed Lyapunov conditions for uniform global asymptotic stability with respect to~\cite[Thm.~3.18]{goebel2012hybrid}. The restriction $\HyR{\real^n\backslash \O}$ is defined in~\eqref{eq:hsRestr}.
\begin{proposition}
\label{prop:recOnotBoundedRelaxed}
Let the hybrid system $\Hy$ in~\eqref{eq:hs} satisfy Assumption~\ref{ass:hbc}, $\O \subset \real^n$ be an open, unbounded set and $V \colon \real^n \to \real_{\ge 0}$ a smooth function, radially unbounded relative to $C \cup D$, strictly positive on $(C \cup D) \backslash\O$ (i.e., $0\notin V((C \cup D) \backslash\O)$), and for which there exist $\mu>0$, $\lambdac\in\real$, $\lambdad\in \real$ such that
\begin{subequations}
\label{eq:LyapFunRecRelaxed}
\begin{align}
& \langle \nabla V(x), f \rangle \!\le\! \lambdac V(x) + \mu \Ind_\O (x) & & \forall x \in C, f \in F(x) \label{eq:LyapFunRecFlowRelaxed} \\
& V(g) \le e^{\lambdad} V(x) + \mu \Ind_\O (x) & & \forall x \in D, g \in G(x). \label{eq:LyapFunRecJumpRelaxed}
\end{align}
Assume further that there exist $\gamma > 0$ and $M >0$ such that, for each maximal solution $\phir$ to $\HyR{\real^n\backslash \O}$, $(t,j) \in \dom \phir$ implies 
\begin{equation}
\label{eq:condOn_tj}
\lambdac t + \lambdad j \le M - \gamma (t+j).
\end{equation}
\end{subequations}
Then $\O$ is \UlyGR for $\Hy$.
\end{proposition}

Intuitively speaking, \eqref{eq:LyapFunRecFlowRelaxed}-\eqref{eq:LyapFunRecJumpRelaxed} allow $V$ to increase even out of $\O$ (for positive $\lambdac$ or $\lambdad$) as long as this increase is balanced by an overall decrease in~\eqref{eq:condOn_tj}, which we emphasize is checked on solutions $\phir$ of the restriction $\HyR{\real^n\backslash  \O}$. 
Proposition~\ref{prop:recOnotBoundedRelaxed} is key for Proposition~\ref{prop:OhRec}.

\subsection{Hybrid system of logic and plant: recurrent set}
\label{sec:recSetHybSys}

At the beginning of Section~\ref{sec:addCT} we specified in~\eqref{eq:hsJustPlant} how solutions generate an observation $o$ corresponding to an open set $Y_o$. We assume:
\begin{assumption}
\label{ass:Yo open}
$Y_o$ is open for each $o \in O$, and $Y_o \cap Y_{o'} = \emptyset$ for each $o, o' \in O$, $o \neq o'$.
\end{assumption}

To design an output feedback scheme to steer $y$ to $Y_o$ through $u$, we make the next classical assumption for setpoint control (see, e.g., \cite[\S 23.6]{hespanha09}) on the plant.
\begin{assumption}
\label{ass:plant}
The number of inputs $m$ is equal to the number of outputs $p$, the matrix $\smat{\Ap & \Bp \\ \Cp & 0}$ is invertible, the pair $(\Ap,\Bp)$ is controllable and the pair $(\Ap,\Cp)$ is observable. 
\end{assumption}

Under Assumption~\ref{ass:plant}, a generic point $y_o \in Y_o$ determines a point $(\xi_o,u_o)$ (used in the feedback law $u$ in~\eqref{eq:controller}) from
\begin{equation}
\label{eq:steady state}
0 = \Ap \xi_o + \Bp u_o, \quad y_o = \Cp \xi_o.
\end{equation}
We can straightforwardly design for~\eqref{eq:ctPlant} the gains $\Kp$ and $\Lp$ of an asymptotically stable output feedback scheme 
\begin{equation}
\label{eq:controller}
\dot{\hat \xi} = \Ap \hat \xi + \Bp u + \Lp (y - \hat y), \, \hat y = \Cp \hat \xi, u = - \Kp (\hat \xi - \xi_o) + u_o
\end{equation}
where $\hat \xi$ is an estimate of $\xi$. Then, for a given observation $o$ we want the solution to generate, we use the scheme in~\eqref{eq:controller} for~\eqref{eq:ctPlant}, select $\rho_o>0$ so that $\Ba(y_o,\rho_o) \subset Y_o$ (which is possible by $y_o \in Y_o$ and $Y_o$ open by Assumption~\ref{ass:Yo open}), and get%
\begin{subequations}
\label{eq:hsGivenObserv}
\begin{align}
&\hspace*{-.2cm} \dot \zeta   = \Fp \zeta \!+\!  \gp, \hspace*{.2cm} \zeta \! \in C_o \! := \! \{ (\xi,\hat \xi)\!\in\! \real^{2 \nu} \!\colon\! \Cp \xi\, \notin \Ba(y_o,\rho_o)^\circ \} \label{eq:hsGivenObservFlow}\\
&\hspace*{-.2cm} \zeta^+  = \zeta , \hspace*{.75cm} \zeta \! \in D_o \! :=\! \{ (\xi,\hat \xi)\!\in\! \real^{2 \nu} \!\colon\! \Cp \xi \!\in\! \Ba(y_o,\rho_o) \}\label{eq:hsGivenObservJump}
\end{align}
\begin{equation}
\zeta \!\!:=\!\! \bmat{ \xi \\ \hat \xi}\!\!, \Fp\!\!:=\!\! \bmat{\Ap & \!\!-\!\Bp \Kp\\ \Lp \Cp & \!\!\!\Ap \!-\! \Bp \Kp\! - \!\Lp \Cp}\!\!, \gp \!\!:=\!\! \bmat{\Bp \Kp \xi_o \!+\! \Bp u_o\!\\  \Bp \Kp \xi_o \!+\! \Bp u_o\! }\!\!.\label{eq:clMatricesPlants}
\end{equation}
\end{subequations}
\eqref{eq:hsGivenObserv} imposes that for a given observation $o$, solutions reach a subset of the corresponding region of interest $Y_o$ before they can jump.
At such a jump, $\xi$ and $\hat \xi$ do not change.

We now augment \eqref{eq:hsGivenObserv} with the BA of the logic. From Section~\ref{sec:BA}, if $s$ and, in particular, $o$ are updated according to~\eqref{eq:HyGBAC}, then the accepting states of the BA $\Au$ are visited infinitely often, which yields words of observations accepted by the LTL formula. 
Then, combining the plant generating observations as in~\eqref{eq:hsGivenObserv} and the controlled BA in~\eqref{eq:HyGBAC} leads to:%
\begin{subequations}
\label{eq:hsCL}
\begin{align}
\xhDot & =
\smat{
\dot \chi\\
\dot \zeta\\
}
=
\smat{
0\\
\Fp \zeta + \gp \\
}
=: \fh(\xh), & &
\xh \in \Ch \label{eq:hsCLflow}\\
\xh^+ & =
\smat{
\chi^+ \\
\zeta^+\\
}
\in
\smat{
\GBAC(\chi) \\
\zeta \\
}
=:\Gh(\xh), & &
\xh \in \Dh.\label{eq:hsCLjump}
\end{align}
The overall state is defined concisely as $\xh := (\chi,\zeta)$ and the overall flow and jump sets $\Ch$ and $\Dh$ are
\begin{align}
\Ch&:= \{\smat{\chi\\\zeta}=\smat{s\\o\\\zeta}\in  \real^{2(\nu+1)} \colon \chi \in \DBAC,\, \zeta \in C_o \} \label{eq:hsCLFlowSet}\\
\Dh&:= \{\smat{\chi\\\zeta}=\smat{s\\o\\\zeta}\in  \real^{2(\nu+1)} \colon \chi \in \DBAC,\, \zeta \in D_o\}, \label{eq:hsCLJumpSet}
\end{align}
\end{subequations}
where $\DBAC$ was defined in~\eqref{eq:DBAC}. Let us comment on~\eqref{eq:hsCL}. $s$ and $o$ do not change during flow in~\eqref{eq:hsCLflow}, and $y=\Cp \xi$ is steered towards $y_o \in Y_o$ by the control law~\eqref{eq:controller} embedded in $\Fp$ and $\gp$. 
$\chi$ is updated according to $\GBAC$ as in Section~\ref{sec:BA}, and $\zeta$ does not change at a jump in~\eqref{eq:hsCLjump}. From~\eqref{eq:hsCLFlowSet}-\eqref{eq:hsCLJumpSet}, jumps are allowed only in the set $\Dh$ comprising all possible $\chi \in \DBAC$ and $\zeta \in D_o$, whereas for all such $\chi$, solutions can only flow before their component $\zeta$ reaches $D_o$.

\begin{remark}
\label{rem:avoidance}
Solutions to~\eqref{eq:hsCL} with the second component equal to some $o$ are allowed to jump only after they reach $D_o$, although they can flow through $Y_{o'}$ with $o' \neq o$.
This adopts the approach of effective path in~\cite[\S III.B]{guo2016communication}. If the LTL semantics imposes active avoidance of $Y_{o'}$
with $o \neq o'$, the tools of this paper can be complemented with hybrid solutions for robust global asymptotic stability of a target in the presence of multiple obstacles as in~\cite{braun2018explicit}. Such an approach involves the intuitive construction of avoidance sets around such $Y_{o'}$'s, and a suitable orchestration between the logical modes of stabilization and avoidance (see also~\cite{braun2018unsafe,berkane2019hybrid}), but is not pursued here due to space constraints.
\end{remark}

\eqref{eq:hsCL} satisfies Assumption~\ref{ass:hbc} (cf. Lemma~\ref{lem:hbc-completeness-GBAC}), and its solutions have the next property.
\begin{lemma}
\label{lem:complSol}
Under Assumptions~\ref{ass:LTLfeas}--\ref{ass:plant}, each maximal solution $\phi$ to~\eqref{eq:hsCL} is complete and $\sup_j \dom \phi =+ \infty$.
\end{lemma}

Thanks to Proposition~\ref{prop:recOnotBoundedRelaxed}, we have the next result.
\begin{proposition}
\label{prop:OhRec}
Under Assumptions~\ref{ass:LTLfeas}--\ref{ass:plant} and with $\rho_o>0$ and $\Ba(y_o,\rho_o) \subset Y_o$ for each $o \in O$, the open, unbounded set
\begin{equation}
\label{eq:Oh}
\Oh \!:= \!\{\xh \!=\!( \chi,\xi,\hat \xi) \!\in\! \real^{2(\nu+1)} \!\colon \!\chi \in \OBAC, \Cp \xi \in Y_o \}
\end{equation}
is \UlyGR for~\eqref{eq:hsCL}.
\end{proposition}

Suppose that for each $o$, the high-level controller given by~\eqref{eq:hsCL} is endowed with a low-level controller that enforces active avoidance of all other regions of interest $Y_{o'}$ with $o' \neq o$, as in Remark~\ref{rem:avoidance}. Under this assumption, for each solution $\phi$ to~\eqref{eq:hsCL}, $\dom \phi$ consists of infinitely many intervals $[t_j,t_{j+1}] \times \{j\}$ (see Section~\ref{sec:prel}) by Lemma~\ref{lem:complSol}.
Hence, for each $\phi$ and each such $j=0,1,\dots$, for some $t_j' \in [t_j, t_{j+1}]$, the output $y$ exits the previous region of interest $Y_{o(t_j,j-1)}$ (if $j>0$) over $[t_j,t_j')\times \{j\}$, belongs to $\real^p \backslash \bigcup_{o' \in O\backslash \{ o(t_j,j)\}} Y_{o'}$ over $[t_j',t_{j+1}]\times \{j\}$ by the previous assumption, and satisfies $y(t_{j+1},j) \in Y_{o(t_j,j)}$.
Moreover, for each $j = 0,1, \dots$, $s$ and $o$ in $\phi$ do not change over $[t_j,t_{j+1}] \times \{ j \}$, so their evolution is captured by $j \!\mapsto\! s(t_j,j)\!=:\! \mathfrak{s}(j)$ and $j \! \mapsto \!o(t_j,j)\!=:\! \mathfrak{o}(j)$. \UGR of $\Oh$ in Proposition~\ref{prop:OhRec} implies then that $\mathfrak s$ reaches infinitely often $\Sf$ and the word $\mathfrak{o}$ is accepted by $\Au$.
The existence of complete solutions
with infinitely many jumps (Lemma~\ref{lem:complSol}) and enjoying recurrence (Proposition~\ref{prop:OhRec}), shows that if the problem is feasible by Assumption~\ref{ass:LTLfeas}, our approach can solve it.

The previous argument shows a limitation of our approach for LTL-synthesis in that it is a high-level controller, and needs to be endowed with a low-level controller for obstacle avoidance whereas it alone solves a relaxed LTL synthesis in terms of effective paths \cite[\S III.B]{guo2016communication}. 

Last, we compare an automata-based solution in \cite{belta2017formal} to ours.%

\begin{remark}
\label{rem:comp cost}
Both our approach and \cite[\S 5.1]{belta2017formal} start from the BA corresponding to the LTL formula.
Denote $v_\Au$ and $e_\Au$ the number of vertices and edges of $\Au$, and $\mathscr{O}(\cdot)$ an asymptotic upper bound in algorithm analysis \cite{cormen2009introduction}.
The overall cost of our approach is $\mathscr{O}(|O| \nu^3 + |S_f|(v_\Au + e_\Au))$ where the first term arises from obtaining $\Kp$, $\Lp$, and each $\xi_o$ and $u_o$ (by solving Lyapunov equations and through matrix operations for \eqref{eq:steady state}-\eqref{eq:controller}), and the second term arises from computing the distances $d$ on $\Au$ for $\deltac$.
On the other hand, we do not build any product automaton of $\Au$ and the transition system discretizing the continuous-time plant. Besides the cost of building such product, we also do \emph{not} have the cost of solving a Rabin game on the graph of the product automaton. This cost is $\mathscr{O}(|X|^2 |S|^2( 1+ |\Sigma|)^2)$ \cite[p.~91]{belta2017formal} where $|X|$ and $|\Sigma|$ are the cardinalities of the set of states and the set of inputs of the transition system. $|X|$ can be very large for the transition system to represent the plant accurately. 
E.g., if each plant state $\xi_i$, $i=1,\dots, \nu$, is discretized into $|\Xi|$ cells, the cost we do \emph{not} have is $\mathscr{O}(|\Xi|^{2\nu} |S|^2( 1+ |\Sigma|)^2)$ where $|\Xi|$ itself can be quite large for an accurate representation of the plant, and we achieve  polynomial instead of exponential complexity in $\nu$.
\end{remark}

\begin{figure}[!t]
\centering
\includegraphics[width=0.48\textwidth]{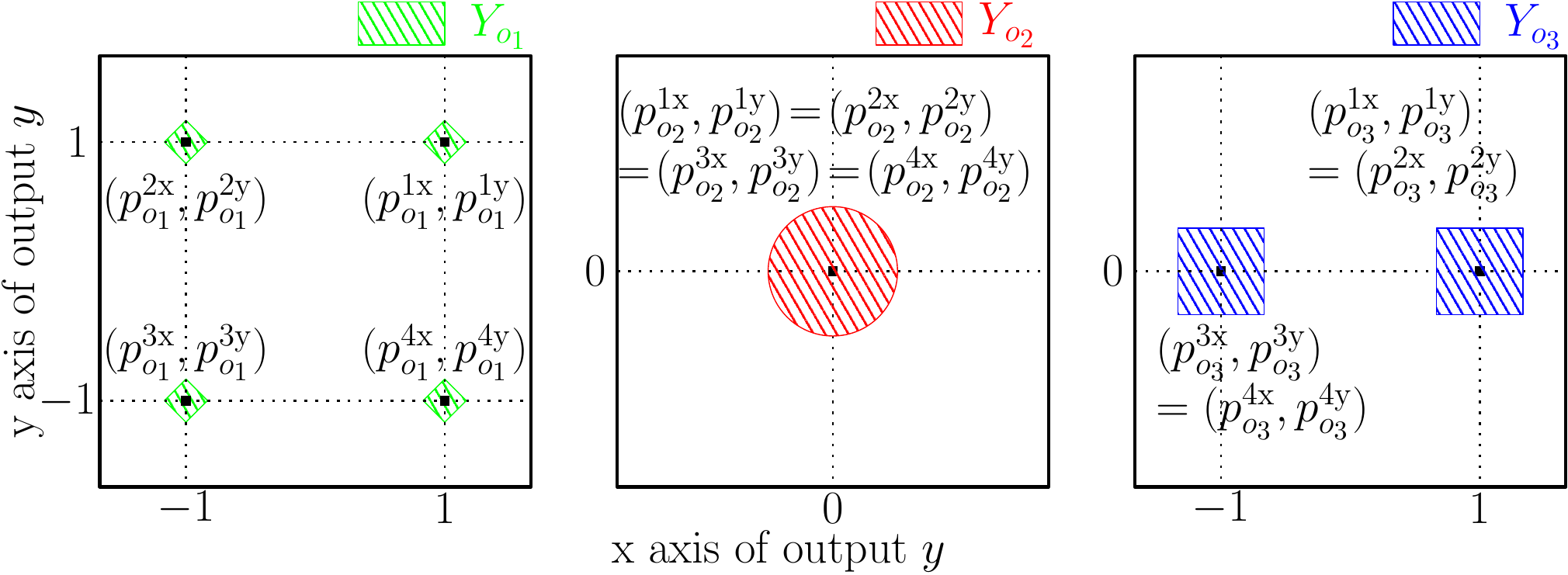}
\caption{The sets $Y_{o_k}$ in~\eqref{eq:Ook} for $k=1,\,2,\,3$ projected for each robot $i=1,\dots,4$ onto the x-y positions of the output $y$.
The  parameters $p^{i \mathrm{x}}_{o_k}$ and $p^{i \mathrm{y}}_{o_k}$ in~\eqref{eq:Ook} take the illustrated values and the radii in~\eqref{eq:Ook} are, for all $i=1,\dots,4$, $r^i_{o_1}=0.1$ and $r^i_{o_2}=0.3$ and $r^i_{o_3}=0.2$.}
\label{fig:Oo}
\end{figure}

\section{Numerical example}
\label{sec:ex}

In this section we illustrate that the control law designed to achieve recurrence of $\Oh$ in~\eqref{eq:Oh} for~\eqref{eq:hsCL} ensures the satisfaction of an LTL formula for service robots.

Each robot $i=1,\dots,4$ has $\mathrm{x}$ and $\mathrm{y}$ positions and velocities $(p^{i\mathrm{x}}, v^{i\mathrm{x}},p^{i\mathrm{y}},v^{i\mathrm{y}})$ as state, $\mathrm{x}$ and $\mathrm{y}$ forces $(f^{i\mathrm{x}}, f^{i\mathrm{y}})$ as input, and only positions as output. It is modeled as a point mass $m_i$ under viscous friction $\gamma_i$. Thus, the state equations are, for $i=1,\dots,4$,
\begin{subequations}
\label{eq:mtx_i}
\begin{align}
\smat{ \dot p^{i\mathrm{x}}\\  \dot v^{i\mathrm{x}} \\ \dot p^{i\mathrm{y}}  \\  \dot v^{i\mathrm{y}} } 
& = 
\smat{0 & 1 & 0 & 0\\ 0 & -\frac{\gamma_i}{m_i} & 0 & 0 \\ 0 & 0 & 0 & 1\\ 0 & 0 & 0 & -\frac{\gamma_i}{m_i} }
\smat{ p^{i\mathrm{x}}\\  v^{i\mathrm{x}} \\ p^{i\mathrm{y}}  \\  v^{i\mathrm{y}} }
+
\smat{0 & 0\\ \frac{1}{m_i} & 0\\ 0 & 0\\ 0 & \frac{1}{m_i}}
\smat{f^{i\mathrm{x}}\\ f^{i\mathrm{y}} } \nonumber \\
& =: \Ap_i \xi^i + \Bp_i u^i
\\
\smat{ p^{i\mathrm{x}}\\  p^{i\mathrm{y}} } 
& = \smat{1 & 0 & 1 & 0} \xi^i =: \Cp_i \xi^i =: y^i,
\end{align}
\end{subequations}
with values $m_i=1$ and $\gamma_i=1$. The overall physical state, input and output are then the stacked vectors $\xi = (\xi^1,\dots,\xi^4)$, $u = (u^1,\dots, u^4)$ and $y=(y^1,\dots,y^4)$. The LTL formula is
\begin{equation}
\label{eq:LTLformula}
(\square \lozenge o_2) \wedge \square\big( o_2 \implies (\bigcirc o_3)\big) \wedge \big( \square \lozenge o_2 \implies \square\lozenge o_1 \big),
\end{equation}
where the symbols $\square$, $\lozenge$, $\bigcirc$ denote respectively the temporal logic operators \textit{always}, \textit{eventually}, \textit{next} as in \cite[Def.~2.2]{belta2017formal}.
An intuitive rendering of the three terms in conjunction is that, in an accepted word $w_O$, \textit{(a)} $o_2$ should be always eventually present, \textit{(b)} whenever $o_2$ is present, $o_3$ should be present next, \textit{(c)} if $o_2$ is always eventually present, then $o_1$ should be always eventually present. 
E.g., $o_1$, $o_2$, $o_3$ can be meaningfully associated respectively with tasks ``charge'', ``pick a parcel'', ``deliver the parcel''. The sets where these tasks are carried out are given in Fig.~\ref{fig:Oo} and defined for $k=1,2,3$ as
\begin{equation}
\label{eq:Ook}
Y_{o_k} \!:=\! \Big\{ y \!\in\! \real^{8} \colon \!\bigg|\!\smat{p^{i\mathrm{x}}-p^{i \mathrm{x}}_{o_k}\\ p^{i\mathrm{y}}-p^{i \mathrm{y}}_{o_k}} \!\bigg|_{n(o_k)} \! < \! r^i_{o_k}, i = 1, \dots, 4 \Big\}
\end{equation}
where all the values of $p^{i \mathrm{x}}_{o_k}$, $p^{i \mathrm{y}}_{o_k}$, $r^i_{o_k}$ are in Fig.~\ref{fig:Oo} and the selected norms are $n(o_1):=1$, $n(o_2):=2$, $n(o_3):=\infty$.

\begin{figure}
\centering
\includegraphics[width=0.47\textwidth]{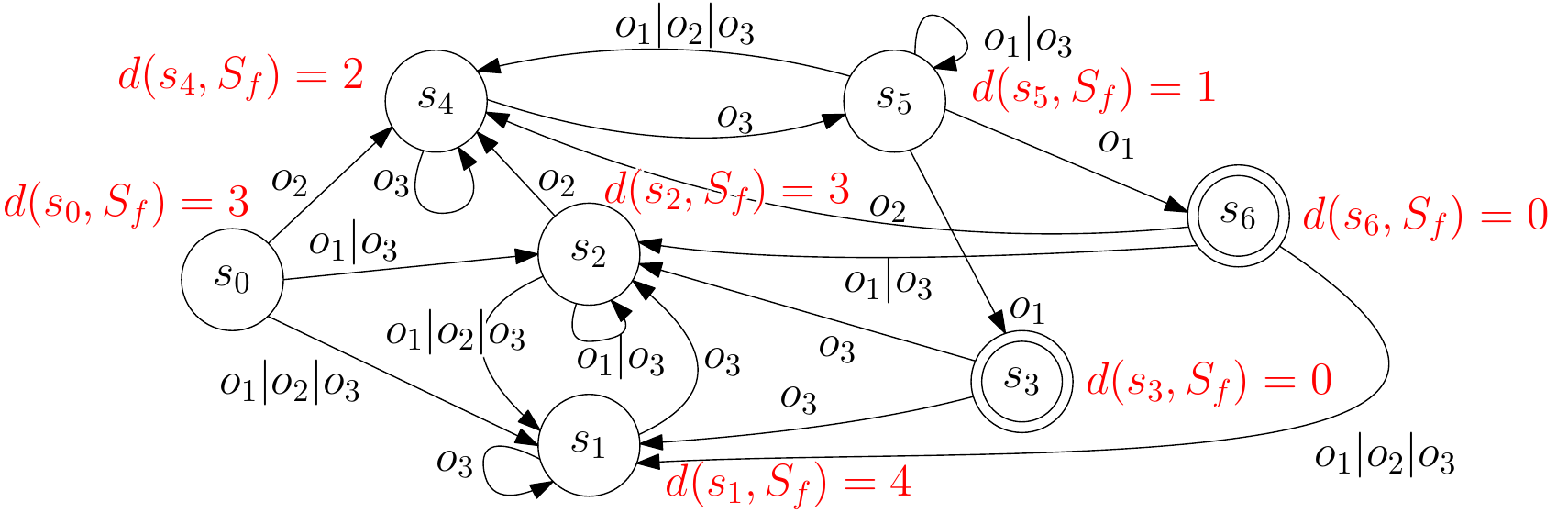}
\caption{The nondeterministic BA corresponding to the formula~\eqref{eq:LTLformula}. The notation $o_i|\dots|o_k$ next to a transition means that such a transition can be taken if either $o_i$, \dots, or $o_k$ are generated. Double circles denote accepting states in $\Sf$. The distance $d$ of each state to $\Sf$ in~\eqref{eq:dist} is labelled in red.
}
\label{fig:automaton}
\end{figure}

We obtain then through the tool \textsc{ltl2ba} \cite{gastin2001fast} the BA corresponding to~\eqref{eq:LTLformula}, partially simplified as in~\cite[Ex.~2.8]{belta2017formal} based on Assumption~\ref{ass:Yo open}. The BA has
\begin{equation}
\label{eq:SOF}
\begin{aligned}
S:=\{ & s_0, s_1, \dots, s_6 \}:=\{ 0, 1, \dots, 6\}, \, S_0:=\{ s_0\}, \,  \\ & \Sf:=\{s_3, s_6\}, \, O:=\{ o_1, o_2, o_3 \}:=\{1,2,3\},
\end{aligned}
\end{equation}
and is depicted in Fig.~\ref{fig:automaton}.
By comparing Fig.~\ref{fig:automaton} with the intuitive rendering of~\eqref{eq:LTLformula} in~\textit{(a)}-\textit{(c)} above, note indeed that for each transition $o_2$ from some $s$ to $s'$, a transition $o_3$ needs to be taken then from $s'$ as per \textit{(b)}; 
to visit infinitely often $s_3$ or $s_6$, $s_4$ needs to be visited infinitely often through the transition $o_2$ and then also the transition $o_1$ to $s_3$ or $s_6$ needs to be taken as per \textit{(a)}, \textit{(c)}.
We assign to each vertex $s\in S$ of the BA the distance $d$ in~\eqref{eq:dist} as in the red labels in Fig.~\ref{fig:automaton}. We report for each $(s,o)$ the set-valued mappings $\deltac(s,o)$ in~\eqref{eq:deltac} and $\GBAC(\smat{s\\o})$ in~\eqref{eq:GBAC} in the next table, where we omit $(s,o)$'s yielding empty $\deltac(s,o)$ and $\GBAC(\smat{s\\o})$.%
\begin{table}[h]
\centering
\[
\begin{array}{lll}
\toprule
(s,o) & \deltac(s,o) & \GBAC((s,o))\\
\midrule
(s_0,o_2), (s_2,o_2), (s_6,o_2) & \{s_4 \} &  \left\{\smat{s_4\\o_3}\right\}\\
(s_1,o_3), (s_3,o_3) & \{s_2 \} &  \left\{\smat{s_2\\o_2}\right\}\\
(s_4,o_3) & \{s_5 \} &  \left\{\smat{s_5\\o_1} \right\}\\
(s_5,o_1) & \{s_3, s_6 \} & \left\{\smat{s_3\\o_3},
\smat{s_6\\o_2} \right\} \\
\bottomrule
\end{array}
\]
\end{table}

\noindent This fully specifies the jump map in~\eqref{eq:hsCLjump}. Next we specify the quantities of the flow map in~\eqref{eq:hsCLflow}. For each $i$ and $o$, define $y^i_o:= (p^{i\mathrm{x}}_o,p^{i\mathrm{y}}_o)$, $\xi^i_o := ( p^{i\mathrm{x}}_o,  v^{i\mathrm{x}}_o, p^{i\mathrm{y}}_o,  v^{i\mathrm{y}}_o )$ and $u^i_o := ( f^{i\mathrm{x}}_o, f^{i\mathrm{y}}_o )$, which satisfy $0 = \Ap_i \xi_o^i + \Bp_i u_o^i$ and $y_o^i = \Cp_i \xi_o^i$ (cf. \eqref{eq:steady state}). 
For $\hat \xi^i := ( \hat p^{i\mathrm{x}}, \hat v^{i\mathrm{x}}, \hat p^{i\mathrm{y}}, \hat v^{i\mathrm{y}} )$ and $\hat y^i := (\hat p^{i\mathrm{x}}, \hat p^{i\mathrm{y}})$, the output feedback scheme~\eqref{eq:controller} for each robot $i=1, \dots, 4$ is 
\begin{equation}
\label{eq:controller_i}
\begin{aligned}
& \dot{\hat \xi}^i = \Ap_i \hat \xi^i + \Bp_i u^i + \Lp_i (y^i - \hat y^i),\,  \Lp_i := \smat{9 & 0\\ 16 & 0\\ 0 & 9\\ 0 & 16}\\
& \hat y^i = \Cp_i \hat \xi^i,  \, \\
& u^i = - \Kp_i (\hat \xi^i - \xi_o^i) + u_o^i, \, \Kp_i := \smat{5/4 & 1 & 0 & 0\\ 0 & 0 & 5/4 & 1}.
\end{aligned}
\end{equation}
With $\hat \xi = (\hat \xi^1,\dots, \hat \xi^4)$, $\hat y = (\hat y^1,\dots, \hat y^4)$, $\xi_o = (\xi_o^1,\dots, \xi_o^4)$ and $u_o=(u_o^1, \dots, u_o^4)$,  it is immediate to obtain from $\Ap_i$, $\Bp_i$, $\Cp_i$ in~\eqref{eq:mtx_i} and $\Lp_i$, $\Kp_i$ in~\eqref{eq:controller_i} the block matrices $\Ap$, $\Bp$, $\Cp$, $\Lp$, $\Kp$ in~\eqref{eq:controller}. The latter fully specify, in turn, the flow map in~\eqref{eq:hsCLflow} by~\eqref{eq:clMatricesPlants}. 
The flow and jump sets in~\eqref{eq:hsCLFlowSet}-\eqref{eq:hsCLJumpSet} are fully specified by $\rho_{o_1}=0.09$, $\rho_{o_2}=0.29$, $\rho_{o_3}=0.19$.

The control design in~\eqref{eq:hsCL} enforces that the output $y$ visits the regions of interest in the order prescribed by the LTL formula in~\eqref{eq:LTLformula}, as shown in Fig.~\ref{fig:path}. The depicted solution corresponds to one of the multiple evolutions encoded by the set-valued mapping $\GBAC$ in the table above, which has observation word $o_2 o_3 (o_1 o_2 o_3)^\omega$ (satisfying \eqref{eq:LTLformula}) and sequence of states $s_0 s_4 (s_5 s_6 s_4)^\omega$ (visiting infinitely often $\Sf$).

The evolution agrees with Lemma~\ref{lem:complSol} and Proposition~\ref{prop:OhRec}, and the acceptance condition of the BA is satisfied.

\section{Conclusions}
In this paper we have related the satisfaction of an LTL formula to the notion of recurrence for hybrid systems. We have first exemplified this relation on the BA corresponding to the LTL formula. Then, in order to address the realistic setting of output feedback, we have extended for open, unbounded sets some Lyapunov-like conditions for recurrence. In particular, one relaxed Lyapunov-like condition has allowed certifying recurrence of a suitable set for the designed hybrid system, formed from LTL formula and linear-time-invariant plant with output feedback. This guarantees satisfaction of the formula by assuming that the high-level controller given by the hybrid system is endowed with a low-level controller for obstacle avoidance, and provides a way for LTL synthesis without relying on discretizations of the plant.

\begin{figure}[!t]
\centering
\hspace*{1.2cm} \includegraphics[width=.75\columnwidth]{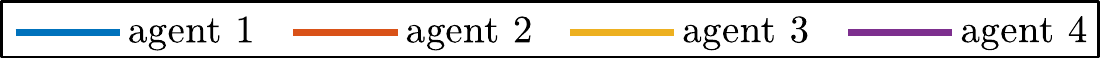}\\
\smallskip
\includegraphics[width=.9\columnwidth]{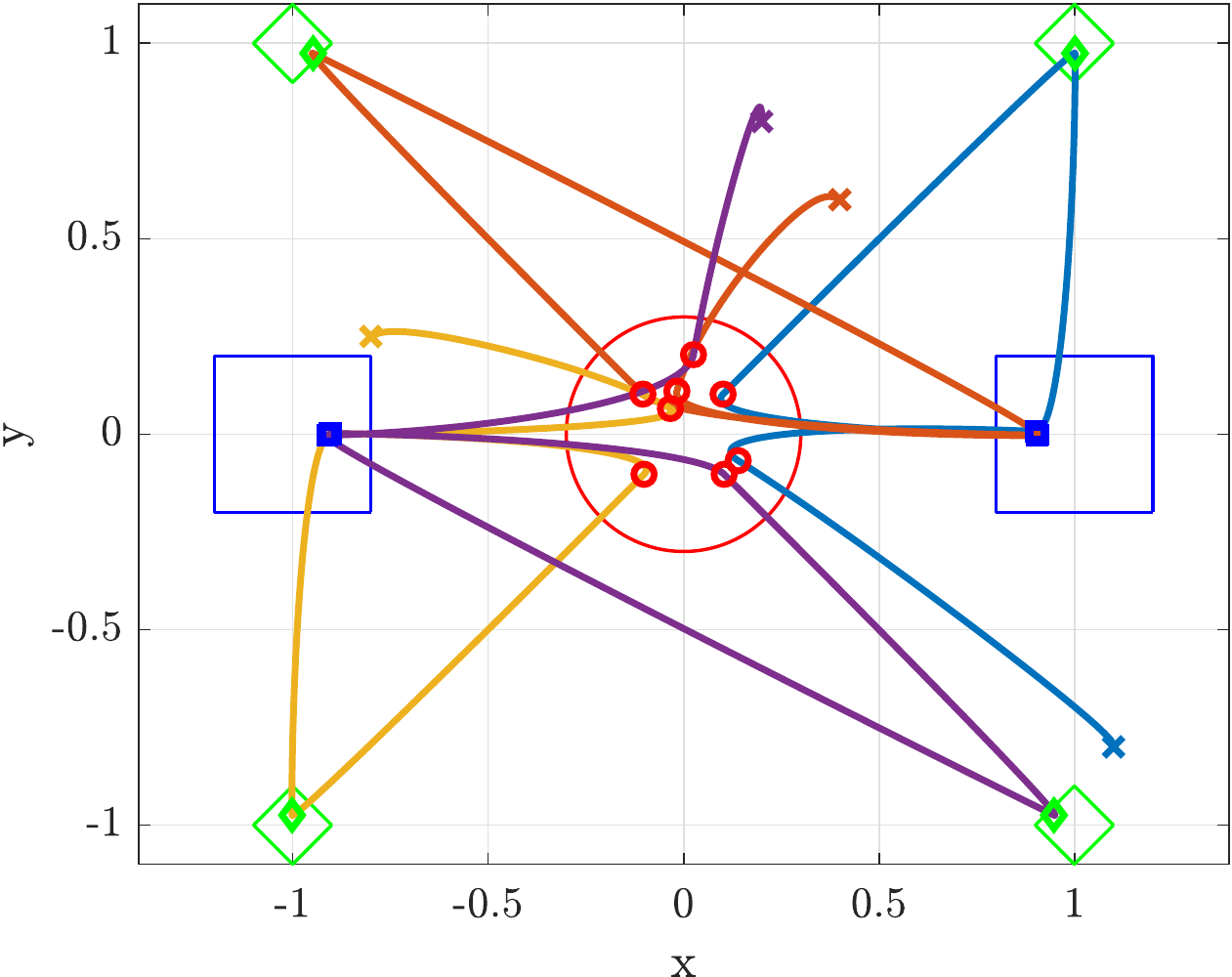}\\
\smallskip
\hspace*{.41cm} \includegraphics[width=.865\columnwidth]{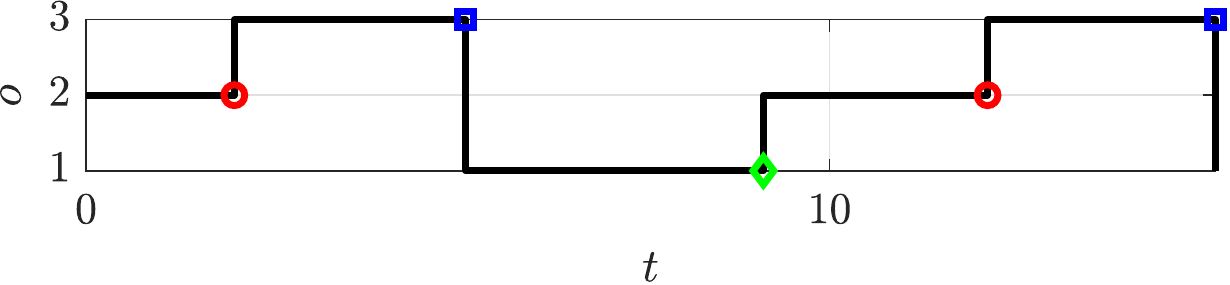} \vspace*{-.3cm}
\caption{
(Top) Evolution of the x and y positions of the robots together with the sets $Y_{o_k}$ of Fig.~\ref{fig:Oo}. Crosses indicate the initial conditions, and diamonds, circles, squares indicate the times when the output reaches $D_{o_1}$, $D_{o_2}$, $D_{o_3}$, respectively. (Bottom) Observations generated by the solution on the top, which are consistent with the order of diamonds, circles, squares encountered on such solution starting from its initial condition. \label{fig:path}}
\end{figure}

\appendix

\textsc{Proof of Lemma~\ref{lem:OCs nonempty}.} Consider $s \in S$ arbitrary in the rest of the proof. By the definition of $O_s$ in~\eqref{eq:Os}, $O_s \neq \emptyset$ in $\OC_s$ in~\eqref{eq:OC_s} by Assumption~\ref{ass:LTLfeas}, otherwise 
$s$ would be removed from $S$ in $\Au$. Then, for some $L \ge 1$, there exist $\delta'_1, \dots, \delta'_L$ all belonging to $\delta(s,O)$ by~\eqref{eq:Os}. 
First, consider $s \in \Sf$ in~\eqref{eq:deltac}. If $\min_{\df \in \delta(s,O)} d(\df,\Sf)= \infty$, there would not be a path from any $\df \in \delta(s,O)$ to any accepting state by~\eqref{eq:distOnAutomaton}, and such $s$ would have been removed from $\Au$. 
So, $\min_{\df \in \delta(s,O)} d(\df,\Sf) < \infty$ and any minimizer among $\delta'_1, \dots, \delta'_L$ is picked. 
Second, consider $s \notin \Sf$ in~\eqref{eq:deltac}. 
Suppose by contradiction that for all $i \in \{1, \dots, L\}$ above, $d(\delta'_i,\Sf) := \min_{\sf \in \Sf} \hat d (\delta'_i, \sf) \ge d(s, \Sf) := \min_{\sf \in \Sf} \hat d(s,\sf)$. 
For all $i \in \{ 1, \dots, L \}$ there exists a shortest path from $\delta'_i$ with $d(\delta'_i, \Sf) := \min_{\sf \in \Sf} \hat d(\delta'_i,\sf) = \hat d(\delta'_i,\sf^i) < \infty$ for some $\sf^i \in \Sf$ (by Assumption~\ref{ass:LTLfeas}) with $\sf^1, \dots, \sf^L$ not necessarily distinct. 
By using such shortest paths and the fact that $\delta'_1, \dots, \delta'_L$ are the only possible successors of $s$, $d(s,\Sf)= \min_{i \in \{1, \dots, L\}} \{ \hat d(\delta'_i,\sf^i) + 1\}=\min_{i \in \{1, \dots, L\}} \{ d(\delta'_i,\Sf) + 1\} \ge d(s,\Sf)+1$. This is a contradiction, hence for $s \notin \Sf$, there exists $k \in \{ 1, \dots , L \}$ such that $d(\delta'_k,\Sf) < d(s,\Sf)$. 
From first and second case, there exists $o \in O_s$ such that $\deltac(s,o) \neq \emptyset$.

\textsc{Proof of Lemma~\ref{lem:hbc-completeness-GBAC}.}
The nontrivial hybrid basic condition to check is $\DBAC \subset \dom \GBAC$, i.e., $\chi \in \DBAC$ implies $\GBAC(\chi) \neq \emptyset$. Indeed, $ \chi = \smat{s\\ o} \in \DBAC$ amounts to $s \in S$ and $o \in \OC_s$ by \eqref{eq:DBAC}. By Lemma~\ref{lem:OCs nonempty}, $\OC_s \neq \emptyset$, i.e., for some $o \in O_s$, $\deltac(s,o) \neq \emptyset$ by~\eqref{eq:OC_s} and there exists $s'\in \deltac(s,o)$. $\OC_{s'} \neq \emptyset$ again by Lemma~\ref{lem:OCs nonempty}, so there exists $o' \in \OC_{s'}$ and $\DBAC \subset \dom \GBAC$ holds. Moreover, $\GBAC(\DBAC) \subset \DBAC$ by construction, see \eqref{eq:GBAC}-\eqref{eq:DBAC}.

\textsc{Proof of Lemma~\ref{lem:recOBAC}.}
Define the Lyapunov function
\begin{equation}
\label{eq:VBA}
\VBA ( \Sso ) := d(s,\Sf),
\end{equation}
whose properties are proven in the next lemma.
\begin{lemma}
\label{lem:VBA}
$\VBA \colon \real^2 \to \real_{\ge 0}$ in~\eqref{eq:VBA} is continuous and there exist $\muBA > 0$ such that
\begin{equation}
\label{eq:Vdecrease}
\begin{aligned}
\VBA(g) - \VBA(\Sso) \le & -1 + \muBA \Ind_{\OBAC}(\Sso) \\
& \forall \Sso \in \DBAC, g \in \GBAC(\Sso).
\end{aligned}
\end{equation}
\end{lemma}
\begin{proof}
Note that $d$ takes integer values. The proof is straightforward by considering separately the cases $\Sso \in \DBAC \backslash \OBAC$ and $\Sso \in \DBAC \cap \OBAC$. For the latter, define
\begin{equation}
\label{eq:dMAX}
\dMAX := \max_{s \in S} d(s,\Sf) < \infty,
\end{equation}
due to Assumption~\ref{ass:LTLfeas}, and use $\muBA := 1 + \dMAX$.
\end{proof}

Since $\HyBAC$ satisfies the hybrid basic conditions by Lemma~\ref{lem:hbc-completeness-GBAC} and $\VBA$ has the properties in Lemma~\ref{lem:VBA} (see \cite[\S 6.2]{subbaraman2016equivalence}), $\OBAC$ is globally recurrent for $\HyBAC$ by~\cite[Thm.~5]{subbaraman2016equivalence} and also \UlyGR, being open and bounded \cite[Prop.~1]{subbaraman2016equivalence}.

\textsc{Proof of Lemma~\ref{lem:recOnotBounded}.}
We prove the two items of Definition~\ref{def:GRandUGR} for \UGR of $\O$ for $\Hy$. 
\newline
\noindent \textit{Item (i) of Definition~\ref{def:GRandUGR} (no finite escape times for~\eqref{eq:hsFlow}).}\newline
Suppose by contradiction that there exists a solution $\phi$ to~\eqref{eq:hsFlow} such that $\bar T:=\sup_t \dom \phi < + \infty$, $\dom \phi=[0,\bar T) \times \{0\}$, and $\lim_{t\to\bar T} |\phi(t,0)|=+\infty$. For such a solution, it holds
\begin{equation}
\label{eq:ValongFiniteEscapeTimes}
V(\phi(t,0)) - V(\phi(0,0)) = \int_0^t\frac{d}{d\tau} V(\phi(\tau,0)) d\tau.
\end{equation}
By taking the limits of the left and right hand sides for $t\to \bar T$, the former diverges to $+\infty$ due to the radial unboundedness of $V$, whereas the latter is upper bounded by $(-1+\mu) t \le |-1+\mu| \bar T$ due to $\mu>0$ in~\eqref{eq:LyapFunRecFlow}. So, such $\phi$ cannot exist.

\noindent\textit{Preliminaries for item (ii) of Definition~\ref{def:GRandUGR}.}\newline
Consider the restriction $\HyR{\RnmO}$ of the hybrid system $\Hy$ as in~\eqref{eq:hsRestr} with the further definitions of flow set $C \cap (\RnmO) =: \Cr$ and jump set $D \cap (\RnmO) =: \Dr$. $\HyR{\RnmO}$ still satisfies Assumption~\ref{ass:hbc} due to $\O$ being open. Let $\phir$ be an arbitrary maximal solution to $\HyR{\RnmO}$. We can consider $\phi(0,0) \in K \cap (\Cr \cup \Dr)$ without loss of generality. Indeed, if $\phi(0,0) \in K\backslash(\Cr \cup \Dr) \subset (K \backslash (C \cup D)) \cup (K \cap \O)$, \UGR is trivially satisfied because $\phi$ has only one point ($\phi \notin C \cup D$) or $\phi(0,0)$ is already in $\O$. For this same argument, item (ii) holds if we prove that to each arbitrary maximal solution $\phir$ to $\HyR{\RnmO}$, it corresponds a solution $\phi$ to $\Hy$ that satisfies item (ii), as we do in the rest of the proof.

\noindent\textit{Item (ii) of Definition~\ref{def:GRandUGR} (uniform times from compact sets).}\newline
By analogous steps to the proof of~\cite[Thm.~3.18]{goebel2012hybrid}, we integrate $V(\phir(\cdot))$ over each interval of flow using~\eqref{eq:LyapFunRecFlow} and compute its increment across each jump using~\eqref{eq:LyapFunRecJump} to obtain
\begin{equation}
\label{eq:boundValongSol}
V(\phir(t,j)) \le V(\phir(0,0)) - (t+j).
\end{equation}
Define the real number
$
\Vu := \sup_{x\in K} V(x) = \max_{x \in K} V(x) \ge 0,
$
where $\Vu \ge 0$ follows from $V$ being smooth and radially unbounded. 
For $(\bar T, \bar J) := \sup \dom \phir$, it must hold $\bar T + \bar J < \hat T := \Vu +1$, in order not to contradict the nonnegativity of $V$ through~\eqref{eq:boundValongSol}. Note that $\hat T$ is uniform over the set $K$.

Let $\xi := \phir(\bar T, \bar J)$ and recall that $\phir$ is an arbitrary maximal solution to $\HyR{\RnmO}$. We have excluded finite escape times in the previous item~(i). $\xi$ cannot belong to $\Cr^\circ$ because $\Cr \subset C \subset \dom F$ by Assumption~\ref{ass:hbc} (indeed, if $\xi \in \Cr^\circ$, there would exist a neighbourhood of $\xi$ such that the tangent cone to $\Cr$ at each point $\xi'$ \cite[Def.~5.12]{goebel2012hybrid} of a neighbourhood of $\xi$ would be $\real^n$ and the intersection with $F(\xi')$ would be nonempty, hence $\phir$ could be extended and would not be maximal \cite[Lemma~5.26(b)]{goebel2012hybrid}). 
Moreover, $\xi \notin \Dr$, otherwise $\phir$ could be extended through a jump. Therefore, $\xi \in \partial \Cr\backslash \Dr$, or, through a jump, $\xi \in \real^n\backslash(\Cr \cup \Dr)$. 
In both cases, when we consider a solution $\phi$ to $\Hy$ with ``initial condition'' $\phi(\bar T, \bar J)= \xi$, such a solution $\phi$ must evolve for some hybrid time in $\real^n\backslash( \Cr \cup \Dr)$ if it evolves, otherwise $\phir$ could be extended from $\xi$ and would not be maximal. 
Since $\real^n\backslash( \Cr \cup \Dr)=\big(\real^n\backslash( C \cup D)\big) \cup \O$, $\phi$ terminates or is in $\O$, which are the two cases of item~(ii) to be shown. Indeed, $\hat T$, which is uniform over the set $K$, gives
$T$ in item~(ii).

\textsc{Proof of Proposition~\ref{prop:recOnotBoundedRelaxed}.}
We follow the proof of Lemma~\ref{lem:recOnotBounded}.\newline
\noindent \textit{Item (i) of Definition~\ref{def:GRandUGR} (no finite escape times for~\eqref{eq:hsFlow}).}\newline
Suppose by contradiction that there exists a solution $\phi$ to~\eqref{eq:hs} such that $\bar T := \sup_t \dom \phi < + \infty$, $\dom \phi = [0,\bar T) \times \{ 0 \}$, and $\lim_{t\to\bar T} |\phi(t,0)| = +\infty$. Note that, due to $\mu >0$ in~\eqref{eq:LyapFunRecFlowRelaxed},
\begin{equation}
\label{eq:flowForFiniteEscapeTimes}
\langle \nabla V(x), f \rangle \le \lambdac V(x) + \mu  \quad \forall x \in C, f \in F(x).
\end{equation}
By defining $t \mapsto v(t) := V(\phi(t,0))$, \eqref{eq:flowForFiniteEscapeTimes} implies that
\begin{equation}
\label{eq:gronwBell}
\dot v(t) \le \lambdac v(t) + \mu.
\end{equation}
By the comparison lemma \cite[Lemma~3.4]{khalil2002nonlinear}, \eqref{eq:gronwBell} implies that
\begin{equation*}
\begin{aligned}
& v(t) \le e^{\lambdac t} v(0) + \tfrac{\mu}{\lambdac} (e^{\lambdac t} - 1) \quad \Longleftrightarrow \\
%\label{eq:gronwBellAlongSol}
& V(\phi(t,0)) \le e^{\lambdac t} V(\phi(0,0)) + \tfrac{\mu}{\lambdac} (e^{\lambdac t} - 1) \quad \forall (t,0)\!\in\! \dom \phi.
\end{aligned}
\end{equation*}
By taking the limits of the left- and right-hand sides for $t \to \bar T$, the former diverges to $+\infty$ due to the radial unboundedness of $V$, whereas the latter is finite. So, such $\phi$ cannot exist.

\noindent\textit{Item (ii) of Definition~\ref{def:GRandUGR} (uniform times from compact sets).}\newline
By analogous steps to the proof of~\cite[Thm.~3.18]{goebel2012hybrid}, we integrate $V(\phir(\cdot))$ over each interval of flow using~\eqref{eq:LyapFunRecFlowRelaxed} and compute its increment across each jump using~\eqref{eq:LyapFunRecJumpRelaxed} to obtain
\begin{equation}
\label{eq:boundValongSolRelax}
V(\phir(t,j)) \le e^{\lambdac t + \lambdad j } V(\phir(0,0)).
\end{equation}
Define the real numbers
\begin{subequations}
\label{eq:boundsV}
\begin{align}
\Vu & := \sup_{x\in K \cap (\Cr \cup \Dr)} V(x) = \max_{x \in K \cap (\Cr \cup \Dr)} V(x) > 0, \label{eq:ubV}\\
\Vl & := \inf_{x\in \Cr \cup \Dr} V(x) = \min_{x \in \Cr \cup \Dr} V(x) > 0, \label{eq:lbV}
\end{align}
\end{subequations}
where $\Vu  \ge \Vl > 0$ follows from $V$ being smooth, radially unbounded, strictly positive on $(C \cup D)\backslash \O$ and $\Cr \cup \Dr$ being a closed set. For $(\bar T, \bar J) := \sup \dom \phir$, it must hold
\begin{equation}
\label{eq:TRecRelax}
\bar T + \bar J \le \hat T := \big(M+\log(\Vu)-\log(\Vl)\big)/\gamma>0.
\end{equation}
Indeed, if \eqref{eq:TRecRelax} was \emph{not} true, i.e., $\bar T + \bar J > \hat T$, we would have
\begin{equation}
\begin{aligned}
& V(\phir(\bar T,\bar J)) \overset{\eqref{eq:boundValongSolRelax}}{\le} e^{\lambdac \bar T + \lambdad \bar J } V(\phir(0,0))\\ & \overset{\eqref{eq:condOn_tj},\, \eqref{eq:ubV}}{\le} e^{M-\gamma(\bar T+ \bar J)} \Vu < e^{M-\gamma \hat T} \Vu \overset{\eqref{eq:TRecRelax}}{=} \Vl,
\end{aligned}
\end{equation}
which is a contradiction (each inequality is obtained thanks to the relationships reported over it). Note that due to~\eqref{eq:TRecRelax} and \eqref{eq:boundsV}, $\hat T$ is uniform over the set $K$. By the same reasoning at the end of the proof of Lemma~\ref{lem:recOnotBounded}, $\hat T$ is then such that item~(ii) of Definition~\ref{def:GRandUGR} is satisfied.

\textsc{Proof of Lemma~\ref{lem:complSol}.}
We apply \cite[Prop.~6.10]{goebel2012hybrid}. Its assumptions are verified since \eqref{eq:hsCL} satisfies both Assumption~\ref{ass:hbc} and a viability condition for each point in $\Ch \backslash  \Dh$. Since finite escape times cannot occur with flow map $\fh$ in~\eqref{eq:hsCLflow}, and $\Gh(\Dh) \! \subset\! \Ch \cup \Dh$, maximal solutions can only be complete. 
Suppose by contradiction that $\sup_j \dom \phi < + \infty$, so this solution $\phi$ stops jumping and $o$ does not change. Since $y_o \! \in \! Y_o$, $\Ap - \Bp \Kp$ and $\Ap - \Lp \Cp$ are Hurwitz and $\xi = \hat \xi = \xi_o$  is the only equilibrium of~\eqref{eq:hsGivenObservFlow} (see \eqref{eq:physFlowInTildeZeta}), $\phi$ also stops flowing at a finite time $t$, which contradicts its completeness we just proved.

\textsc{Proof of Proposition~\ref{prop:OhRec}.}
We verify that the assumptions of Proposition~\ref{prop:recOnotBoundedRelaxed} hold, and this concludes \UGR of $\Oh$ in~\eqref{eq:Oh}. \eqref{eq:hsCL} satisfies Assumption~\ref{ass:hbc} and we take $\Vh$ in~\eqref{eq:Lf} below as the Lyapunov function used in Proposition~\ref{prop:recOnotBoundedRelaxed}. Note that $\dot \zeta = \Fp \zeta + \gp$ in~\eqref{eq:hsCLflow} can be written in the error variables
\begin{equation}
\label{eq:physFlowInTildeZeta}
\tilde \zeta := \smat{\xi - \xi_o\\ \xi - \hat \xi} \text{ as }
\dot{\tilde \zeta} = \bmat{ \Ap - \Bp \Kp & \Bp \Kp\\ 0 & \Ap - \Lp \Cp} \tilde{\zeta} =: \tilde{\Fp} \tilde{\zeta}.
\end{equation}
Under Assumption~\ref{ass:plant}, $\Ap - \Bp \Kp$ and $\Ap - \Lp \Cp$ are selected Hurwitz, so that $\tilde{\Fp}$ is Hurwitz as well.  For $\Qp = \Qp^T >0$, $\Pp = \Pp^T>0$ is then the unique solution to the Lyapunov equation
$
\Pp \tilde{\Fp} + \tilde{\Fp}^T \Pp = - \Qp.
$
\begin{equation}
\label{eq:Lf_zeta}
W(\zeta,o) := \smat{\xi - \xi_o\\ \xi - \hat \xi}^T \Pp \smat{\xi - \xi_o\\ \xi - \hat \xi} = \tilde \zeta^T \Pp \tilde \zeta
\end{equation}
is a Lyapunov function for the point $\zeta =( \xi_o, \xi_o)$. Define the quantities in the next table, where the generic $\xi_o$ is as in~\eqref{eq:steady state},  $\lambda_\textup{min}[\cdot]$ and $\lambda_\textup{MAX}[\cdot]$ denote the minimum and maximum eigenvalue of the argument matrix, and $\dMAX$ is as in~\eqref{eq:dMAX}.
\begin{table}[h]
\centering
\begin{tabular}{c}
\toprule
Definition of quantities used in the rest of the proof\\
\midrule
{
\begin{minipage}{0.935\linewidth}
\vspace*{-.23cm}
\begin{align}
\wmin & := \min_{o\in O,\,(\xi, \hat \xi) \in C_o}  \smat{\xi - \xi_o\\ \xi - \hat \xi}^T \Pp \smat{\xi - \xi_o\\ \xi - \hat \xi} > 0 \label{eq:wmin}\\
J_1 & := \max_{o,\,o' \in O} \smat{\xi_o - \xi_{o'}\\0}^T \Pp \smat{\xi_o - \xi_{o'}\\0} > 0 \label{eq:J_1}\\
\lambda'& :=\lambda_\textup{min}[\Qp]/\lambda_\textup{MAX}[\Pp]>0 \label{eq:lambdaPrime}\\
\lambda & \in \Big( \frac{1-\theta}{\theta} \frac{\dMAX}{\wmin},\frac{1}{J_1} \Big) \text{ for } \theta \in (0,1) \label{eq:lambda}
\end{align}
\end{minipage}
}\\
\bottomrule
\end{tabular}
\end{table}
\newline For these quantities we note: \textit{(i)}~$\wmin>0$ since $\Pp>0$ and $\wmin =0$ would imply, for some $o$, $y=y_o$, which is impossible for $(\xi, \hat \xi) \in C_o$; \textit{(ii)}~$J_1>0$ since $J_1=0$ would imply $\xi_o = \xi_{o'}$ for $o \neq o'$ and, in turn, $y_o = y_{o'}$ for $o \neq o'$, which is impossible by the disjointness of $Y_o$ and $Y_{o'}$ in Assumption~\ref{ass:Yo open}; \textit{(iii)}~$\lambda'>0$ follows from $\Pp>0$ and $\Qp>0$; \textit{(iv)}~the interval for $\lambda$ is well-defined because $1/J_1>0$, and for the given $\dMAX/\wmin>0$, we can always select $\theta\in (0,1)$ close enough to $1$ 
so that $0 < \frac{1-\theta}{\theta} \frac{\dMAX}{\wmin} < \frac{1}{J_1}$. With $\VBA$ in \eqref{eq:VBA}, $\lambda>0$ in~\eqref{eq:lambda}, and $W$ in~\eqref{eq:Lf_zeta}, the Lyapunov function is
\begin{equation}
\label{eq:Lf}
\Vh(\xh) := \VBA(\chi) + \lambda W(\zeta,o).
\end{equation}
$\Vh$ is smooth ($\VBA$ and $W$ are smooth), radially unbounded relative to $\Ch \cup \Dh$ and strictly positive on $(\Ch \cup \Dh)\backslash \Oh$ (since $\VBA$ vanishes only for $s\in \Sf$ and $W$ only for $\xi\! =\! \hat \xi\!=\! \xi_o$).

\noindent\textit{Flow condition~\eqref{eq:LyapFunRecFlowRelaxed}.}
We have that for each $\xh \in \Ch$
\begin{equation}
\label{eq:VhDot1}
\begin{aligned}
& \langle \nabla \Vh (\xh), \fh(\xh) \rangle  = - \lambda \tilde \zeta^T \Qp \tilde \zeta \! \le  \!  - \lambda \lambda' \tilde \zeta^T \Pp \tilde \zeta\\
& \quad = - \lambda' (1-\theta) \Vh(\xh) - \lambda' \big( \theta \Vh(\xh) - \VBA(\chi) \big),
\end{aligned}
\end{equation}
where the first equality follows from \eqref{eq:hsCLflow}, \eqref{eq:Lf_zeta} and \eqref{eq:physFlowInTildeZeta} (where $\tilde \zeta$ is defined), the inequality from~\eqref{eq:lambdaPrime}, and the second equality from simple computations. 
\eqref{eq:Lf},$\,$\eqref{eq:wmin} and $\VBA(\chi) \!\le\! \dMAX$ (for each $\chi$) imply $\big( \theta \Vh(\xh) \!-\! \VBA(\chi) \big) \!\ge\! \theta \lambda \wmin \!-\! (1-\theta) \dMAX\!>\!0$ due to the lower bound of $\lambda$ in~\eqref{eq:lambda}.
So, 
\begin{equation}
\label{eq:VhDot2}
\langle \nabla \Vh (\xh), \fh(\xh) \rangle \! \le \! - \lambda' (1\!-\!\theta) \Vh(\xh)=: \! \lambdac \Vh(\xh)
\end{equation}
proves \eqref{eq:LyapFunRecFlowRelaxed} with $\lambdac<0$ (recall $\theta \in (0,1)$ in~\eqref{eq:lambda}).

\noindent\textit{Jump condition~\eqref{eq:LyapFunRecJumpRelaxed}.}
We use in this step that, for each $a$, $b$, and $\Pp=\Pp^T >0$ of compatible dimensions,
\begin{equation}
\label{eq:boundPosDefP}
(a+b)^T\Pp(a+b) \le 2 a^T \Pp a + 2 b^T \Pp b.
\end{equation}
For $\xh^+=(\chi^+,\zeta^+)=(s^+,o^+,\xi^+,\hat \xi^+)=(s^+,o^+,\xi,\hat \xi)$,
\begin{equation}
\label{eq:VhPlusGen}
\begin{aligned}
\Vh(\xh^+) &=\VBA(\chi^+) + \lambda W(\zeta,o^+)  \\
& \le \VBA(\chi^+) + \lambda \big( 2 W(\zeta,o) + 2 J_1\big)
\end{aligned}
\end{equation}
by~\eqref{eq:boundPosDefP} and then \eqref{eq:J_1}. In the case $\xh \in \Dh \backslash \Oh$, $s \notin \Sf$, $\VBA(\chi) \ge 1$ and $\VBA(\chi^+) \le \VBA(\chi) -1$, so that \eqref{eq:VhPlusGen} becomes
\begin{equation}
\label{eq:VhPlusNotOh}
\Vh(\xh^+)\! \le\! 2 \big(\VBA(\chi)\! +\! \lambda W(\zeta,o) \big) \!+\! 2 \lambda J_1 \!-\!1 \!-\! \VBA(\chi) \! \le \! 2 \Vh(\xh),
\end{equation}
where $\VBA(\chi) \ge 1$ and the upper bound of $\lambda$ in~\eqref{eq:lambda} yield the second inequality. In the case $\xh \in \Dh \cap  \Oh$, $s \in \Sf$ and $\VBA(\chi) = 0$, so that \eqref{eq:VhPlusGen} becomes
\begin{equation}
\label{eq:VhPlusOh}
\Vh(\xh^+)\! \le\! 2 \VBA(\chi) + \dMAX + \lambda \big( 2 W\!(\zeta,o) + 2 J_1\big) \! =:\! 2 \Vh(\xh) + \muh.
\end{equation}
By \eqref{eq:VhPlusNotOh} and \eqref{eq:VhPlusOh}, \eqref{eq:LyapFunRecJumpRelaxed} is proven with $\lambdad = \log 2$ and $\mu=\muh$.

\noindent\textit{Restriction $\HyR{\real^{2(\nu +1)}\backslash \Oh}$ and \eqref{eq:condOn_tj}.}
By construction of the jump map for $s$ in~\eqref{eq:hsCLjump} (see \eqref{eq:GBAC} and \eqref{eq:deltac}), the distance $d$ to the accepting states decreases at each jump by at least $1$ (since $\Dh \cap (\real^{2(\nu +1)}\backslash \Oh)$ excludes $s \in \Sf$) and is upper bounded by $\dMAX$. Then, each solution $\phir$ to $\HyR{\real^{2(\nu +1)}\backslash \Oh}$ with $(t,j) \in \dom \phir$, has $j \le \dMAX$. From the steps above for flow/jump we have $\lambdac<0$ and $\lambdad >0$. Define $M:= (\lambdad - \lambdac) \dMAX>0$ and $\gamma:= - \lambdac >0$. Then, \eqref{eq:condOn_tj} is proven by
\begin{equation}
\lambdac t + \lambdad j \le \lambdac (t + j) + (\lambdad - \lambdac) \dMAX = M - \gamma (t+j).
\end{equation}

\bibliographystyle{IEEEtranS}
\bibliography{IEEEabrv,refs}

\end{document}